\documentclass[12pt, leqno]{amsart}
\usepackage{amssymb, amsmath}
\usepackage{latexsym}
\usepackage{graphics}
\usepackage{color}
\usepackage{url}
\usepackage{amsthm}

\numberwithin{equation}{section}
\newtheorem{theorem}[subsection]{Theorem}

\newtheorem{remark}[subsection]{Remark}
\newtheorem{lemma}[subsection]{Lemma}
\newtheorem{corollary}[subsection]{Corollary}

\newtheorem{definition}[subsection]{Definition}

\newcommand{\Z}{\mathbb Z}

\title{Computing $\varphi(N)$ for an RSA module with a single Quantum Query}

\author{Luis Víctor Dieulefait }

\address{Departament de Matem\`atiques i Inform\`atica,
Universitat de Barcelona, Barcelona, Spain. email: ldieulefait@ub.edu}

\author{Jorge Urroz}

\address{Departamento de Matemáticas e Informática aplicadas a la Ingeniería Civil y Naval, Universidad Politécnica de Madrid, Madrid, Spain.}
\thanks{Affiliations are included at the end of the paper. The first named author is supported by the grant PID2022-136944NB-I00 of the Ministerio de Ciencia, Innovaci\'{o}n y Universidades (Spain).}

\begin{document}

\begin{abstract}
    In this paper we give a polynomial time algorithm to compute $\varphi(N)$ for an RSA module $N$ using as input the order modulo $N$ of a randomly chosen integer.  This provides  a new insight in the very important problem of factoring an RSA module with extra information. In fact, the algorithm is extremely simple and  consists only on a computation of a greatest common divisor, two multiplications and a division. The algorithm works with a probability of at least 
    $1-\frac{1}{N^{1/2-\epsilon}}$, where $\epsilon$ is any small positive  constant. 
    
\end{abstract}

\maketitle

\section{Introduction}

In this paper we give a polynomial time algorithm to compute $\varphi(N)$ for an RSA module $N$ using as input the order modulo $N$ of a randomly chosen integer. This gives a new insight in the very important problem of factoring an RSA module with extra information. In fact, the algorithm is extremely simple and it consists only on a computation of a greatest common divisor, two multiplications and a division (see Theorems \ref{th:principal} and \ref{th:squareroot}). As it is well-known, from this value it is easy to factor $N$ by just solving a quadratic equation. The algorithm works with a probability greater than $1-\frac{1}{N^{1/2-\epsilon}}$, where $\epsilon$ is any small positive  constant, for $N$ large enough depending on $\epsilon$, and the input can be obtained from a quantum computer by a single application of Shor's algorithm (cf. \cite{shor1} and \cite{shor2}), assuming that the output of that algorithm is exactly the order module $N$ of a randomly chosen integer. This assumption is reasonable since, as explained in \cite{eke}, the probability that Shor's algorithm fails to give the right order is negligible, under suitable conditions on the parametrization and post-processing of the quantum algorithm (cf. \cite{eke}, section 2.1).  \\

\
There is a vast literature related to the problem of factoring via the knowledge of the order of an element. We recommend the reading of the paper \cite{eke} by M. Ekera, in which a detailed description of  different techniques and advances in the literature are shown, including modifications of the part of Shor's order finding algorithm (see \cite{sei}) or improving the lower bounds on the success probability (see \cite{knill}), or improvements in special cases like odd orders (see \cite{johnston,law,many}) or special integers (see \cite{eke2,ekehas, glms, xu}), among others. The paper \cite{eke} also includes a nice view on some previous  results about the extremely important factorization problem.

\

The main contribution of  \cite{eke} is  another polynomial time algorithm is given that factors an integer taking as input the order of a random integer, which is based on a variation of a method of V. Miller (the same method used by Shor to end his factorization algorithm, but this time with the extra clever idea of moving the base in a post-processing process in a classical computer). The algorithm in \cite{eke} can be used to factor arbitrary integers, not only RSA modules, but even for the case of RSA modules the probability of success is smaller than $1-\frac 1{C \log^2 N}$ (see \cite{eke}, Theorem 1 \footnote{we omit here the term $-1/{2^k}$, where $k$ denotes the number of iterations, since this term is dominated by the other, for $k$ sufficiently large}), thus the probability of failure is exponentially larger than the one we obtain. 
One could argue that  the algorithm in \cite{eke} is  sufficiently well-suited for practical purposes since the probability of failure is rather small, however, if we want to use the algorithm repeatedly, then since
$$ (1-\frac{1}{c(\log N)^2})^{(c\log N)^2} \approx e^{-1} $$
it gives a probability of getting at least one failure of about $1-\frac{1}{e}$ when repeating the algorithm a polynomial (quadratic) number of times. On the other hand,  our algorithm can be used an exponential number of times and still ensures a probability of failure exponentially small in the parameters. 

Also, a key advantage of the present algorithm is its extreme simplicity, very easy to implement. Since it is nothing but an application of Euclids' algorithm, two multiplications and one division, one could  get quasilinear complexity, see \cite{multi} and \cite{mo}. More precisely, the order of the algorithm is $O(\log N \cdot (\log \log N)^a)$ where the exponent $a$ is any real number larger than $2$. \\

\

The possibility of computing $\varphi(N)$ for an RSA module given the order of a random integer was also observed in \cite{glms} but only for the case of safe semiprimes (i.e., modules of the form $p \cdot q$ where both $p$ and $q$ are safe primes, which are by definition primes of the form $2 \cdot t + 1$ for another prime $t$), whereas our algorithm works  without this restriction.\\

\

To finalize this introduction we would like to  mention an additional  result that comes from our method. Concretely, in $1978$, Rivest, Shamir and Adleman  in \cite{rsa-factd} gave a probabilistic  algorithm proving that factoring $N$ is equivalent to knowing the secret key on the RSA cryptosystem. Namely knowing the public information $N$ and the public exponent $e$, if one can compute $d$, the inverse of $e$ modulo $\varphi(N)$, then one can factor $N$ in polynomial time with high probability. Then May and Coron in \cite{maycoron} improved the result showing a deterministic algorithm for the reduction of the knowledge of $d$ to factor $N$. For that, they use an improvement of Coppersmith algorithm (cf. \cite{cop}) valid not only for the integer of unknown factorization, but  any divisor of it. Here we show a straightforward proof of this equivalence, in the case that the exponent $e$ is not too large. Recall that, in practice, the most common exponent is the constant $3$ or $65537$. 

\subsection{The algorithm} The idea behind the proof of the main result of this paper, Theorem \ref{th:principal}, is very simple. Recall that the objective is to find $\varphi(N)$. Then, we also can factor $N$ by solving a system of two equations, as explained at the beginning  of in Section \ref{sec:factoring}.  The information that we have is $N$ and $x$, the order of an element in $(\Z/N \Z)^*$. Hence, $x$ is a divisor of the unknown $\varphi(N)$ and  there are two possibilities. If $x$ is big enough
we will find $\varphi(N)$ by noting that,  for primes $p,q$ of the same magnitude, the interval given by $N-C\sqrt N,N+C\sqrt N$, for some constant $C$,  contains exactly one multiply of $x$ which must be $\varphi(N)$, so we can find $\varphi(N)$ with elementary multiplications and divisions. On the other hand we prove that divisors of $\varphi(N)$ which are small appear with very low probability. As we will see a key feature to deduce this comes from our technique to enlarge divisors  proving that whenever $x$ is the order of an element in $(\Z/N\Z)^*$ then   $D=x(x, N-1)$ divides $\varphi(N)$.

\

In this way, Section \ref{sec:factoring} contains the technical results to take care of large divisors, while Section \ref{sec:order} deals with small divisors. For that we need to give an upper bound in the probability to find small orders and hence, we need a complete knowledge of the function counting how many elements are there having as order a concrete divisor. The section gives a fully explicit formula for that function and, once we have it, we can get our main theorem and other applications as a corollary.  Finally, we include Section \ref{sec:alg} containing a very simple description of the algorithm with  its complexity.

\

 To simplify the presentation, instead of working with an arbitrary RSA module, we will restrict to the case of a module $N$ such that the two prime factors $p$ and $q$ have the same number of bits, as we have mentioned above. This is the case for RSA modules in practice, in agreement with the Digital Signature Standards fixed by the National Institute of Standards and Technology (see \cite{dss}, Appendix A.1.1). 

\

We remark that our results can be easily adapted to cover also the case of arbitrary RSA modules. In fact, if we assume that both prime factors of $N$ are larger than $N^{1/4}$ then a result similar to Theorem \ref{th:principal} can be proved, just changing $1-\frac{1}{N^{1/2 - \epsilon}}$ by $1-\frac{1}{N^{1/4 - \epsilon}}$ in the probability of success. On the other hand, there is no need to consider the case where one of the prime factors of $N$ is smaller than $N^{1/4}$, since as it is well-known in that case the method of Coppersmith (see \cite{cop}) gives a polynomial time deterministic algorithm to factor $N$ (with a classical computer). \\

The following notation will  be used through the whole paper. Given two integers $a,b$, $(a,b)$ will be their greatest common divisor, while $[a,b]$ will be their least common multiple. It will be useful to recall that 
 \begin{equation}\label{eq:gl}
 ab=[a,b](a,b)
 \end{equation}
  Also, $N=pq$ will be the product of the two  prime numbers $p$ and $q$ such that  
\begin{eqnarray}\label{eq:factorpq}
\nonumber p-1&=&\prod_{i=1}^rp_i^{a_i}\\
q-1&=&\prod_{i=1}^rp_i^{b_i}
\end{eqnarray}
where a prime $l|(p-1,q-1)$ if and only if $l=p_i$ for some $1\le i\le r$. We will denote $m_i=\min\{a_i,b_i\}$ and $M_i=\max\{a_i,b_i\}$ for $i=1,\dots ,r$.

\section{Explicit formulas for $\varphi(N)$ and factoring}\label{sec:factoring}
It is well known  that an RSA modulus $N=pq$ can be factored knowing $\varphi(N)$, simply solving a system of two equations in the unknown $p,q$.  Namely $pq=N$ and $p+q=N-\varphi(N)+1$ give $p$ and $q$ the two solutions of the equation $x^2-Bx+N=0$, for $B=(N-\varphi(N)+1)$ or in other words
$$
q=\frac{B+\sqrt{B^2-4N}}{2},\qquad p=\frac{B-\sqrt{B^2-4N}}{2}.
$$
Since square rooting is as difficult as multiplication (cf. \cite{alt}), it follows from the multiplication algorithm in \cite{multi} that knowing $\varphi(N)$, we can find the factors  $p,q$ in $O(\log N  \cdot \log\log N )$ bit operations. For that reason our algorithms will simply analyze the cost of computing $\varphi(N)$. But, in fact, we now prove that it is enough to know either a big enough factor of $\varphi(N)$, or a small enough multiple of $\varphi(N)$.  

\

We start by showing that a large enough divisor $D|\varphi(N)$ have information enough to find $\varphi(N)$. The idea is simply to observe that $\varphi(N)$ and $N$ are in an interval of length $\sqrt N$, and in this interval  if $D$ is large enough, only one multiple of $D$ will fit, which must be $\varphi(N)$.

\begin{theorem}\label{th:squareroot} Let $N=pq$ be the product of two unknown prime numbers $p,q$ of at most $L$ bits, say $p<q\le2^{L}$.
Suppose we know $D$ such that $DX=(p-1)(q-1)$ and $D\ge 2^{L+1}$. Then, we can find $p,q$ in polynomial time. 
\end{theorem}
{\bf Proof.} 
$$
DX=(p-1)(q-1)=N-(p+q)+1,
$$
so 
$$
\frac{N+1}{D}=X+\frac{p+q}{D},
$$
where 
$$
0<\frac{p+q}{D}<1,
$$
so 
$$
\left[\frac{N+1}{D}\right]=X,
$$
and once we know $X$, then we know $(p-1)(q-1)$ and we solve the system
\begin{eqnarray*}
&&N=pq\\
&&N+1-DX=p+q.
\end{eqnarray*}
It is important to emphasize that the previous theorem gives us a new formula for $\varphi(N)$ which we will highlight for independent interest.
\begin{corollary} \label{th:laphi} Let $N=pq$ be the product of two unknown prime numbers $p,q$ of at most $L$ bits, say $p<q\le2^{L}$.
Suppose we know $D$ such that $DX=(p-1)(q-1)$ and $D\ge 2^{L+1}$. Then, 
$$
\varphi(N)=\left[\frac{N+1}{D}\right]D.
$$
\end{corollary}
We isolate now and interesting case  when the large divisor of $\varphi(N)$ is  $D=(p-1,q-1)$.
\begin{theorem}\label{th:gcd} Let $N=pq$ be the product of two unknown prime numbers $p,q$  such that $p+q<D^2$, where $D=(p-1,q-1)$ is known. Then we can factor $N$ in polynomial time.
\end{theorem}
{\bf Proof.} We know that for some integers $R_p,R_q$ 
\begin{eqnarray*}
p&=&1+R_pD\\
q&=&1+R_qD,
\end{eqnarray*}
where $R_q+R_p=\frac{q-1}{D}+\frac{p-1}{D}<D$ by our conditions. Then, since
$$
N=1+(R_q+R_p)D+R_qR_pD^2
$$
we get 
$$
R_p+R_q=\frac{N-1}D\pmod D,
$$
in particular reducing $(N-1)/D$ modulo $D$ we get $R_p+R_q$ and, from there
 $p+q=2+(R_p+R_q)D$, and we factor $N$ by solving the system in $p+q$ and $pq$ as before.

\begin{remark} In the context of the paper one case of interest could be when $p,q$ are prime numbers of the same number of bits, say in the interval  $2^{L-1}<p<q<2^L$ and the gcd has half of them, namely $(p-1,q-1)>2^{(L+1)/2}$.
\end{remark}

Our next result is about  the complementary problem of knowing multiples of $\varphi(N)$. It has a direct implication with 
the mentioned  relation between factoring and finding the secret key on an RSA cryptosystem.

\begin{theorem} Let $N=pq$ be the product of two prime numbers  of  $L$ bits $2^{L-1}<p<q<2^{L}$, with  $(e,\varphi(N))=1$, $e<\frac{\sqrt2}3\sqrt N$ and suppose we know its inverse $d$ modulo $\varphi (N)$, namely $1\le d\le\varphi(N)$, such that $ed\equiv 1\pmod {\varphi(N)}$. Then, we can factor $N$  performing just one multiplication and two divisions.
\end{theorem}
 In fact, we prove a bit more. If we let $M=ed-1=k\varphi(N)$ we have the following explicit formula of independent interest
\begin{corollary} Let $N=pq$ be the product of two prime numbers of  $L$ bits $2^{L-1}<p<q<2^{L}$ $(e,\varphi(N))=1$, $e<\frac{\sqrt2}3\sqrt N$ and suppose we know its inverse $d$ modulo $\varphi (N)$, namely $1\le d\le\varphi(N)$. Then, denoting  $M=ed-1$ we have
$$
\varphi(N)=\frac{M}{\left[\frac{M}{N}\right]+1}.
$$
\end{corollary}
\begin{proof} (Of the theorem and the corollary) We know that  
$$
M=ed-1=k\varphi(N)=k(N-(p+q)+1),
$$
and dividing by $N$ we get
$$
\frac{M}{N}=k-\frac{k(p+q-1)}{N}.
$$
but since $M<ed<c\sqrt {N}\varphi(N)$  for $c=\frac{\sqrt2}3$, we have $k<c\sqrt N$. On the other hand, since  the function $f(x)=x+\frac1x$  is increasing for $x>1$, we have $f(\sqrt{q/p})<f(\sqrt 2)=1/c$, for any $q/p<2$ as in the hypothesis. From there we get   $p+q<\sqrt{pq}/c$, and then $k(p+q-1)<k(p+q)<N$. Hence
$$
0<\frac{k(p+q-1)}{N}<1,
$$
which gives $k=\left[\frac{M}{N}\right]+1$, and hence the result in the corollary. Once we know $\varphi(N)$ we can factor $N$ as usual by solving the system in $p+q$ and $pq$.
\end{proof}

\section{Factoring and the order of elements modulo $N$.}\label{sec:order}

The results in the previous section are directly related with the problem of factoring $N$ knowing the order of an element in $(\Z/N\Z)^*$, an information that can be obtained from  Shor's algorithm. Concerning this problem, in \cite{eke} it is proved that knowing the order of a randomly chosen element modulo $N$ one can factor $N$ with probability smaller than $1-\frac 1{C \log^2 N}$ by applying a procedure that is a variation of the method of Miller proposed by Shor himself in this context, but with the novel feature that the base is allowed to vary while a smooth factor is added to the order. On the other hand, in \cite{glms} the authors show that $\varphi(N)$ can be computed for an RSA module from the order of an element modulo $N$, with high probability, but only for the case of safe semiprimes. In the remaining of this section, we improve dramatically the previous results and prove that the knowledge of the order of an element in $(\Z/N\Z)^*$ gives the factorization of $N$ with probability $1-\frac{1}{N^{1/2-\epsilon}}$ for arbitrary $\epsilon>0$ and large enough $N$,  assuming that the prime factors of the RSA modulus $N$ have the same length (also, as remarked in the introduction, a variant of this algorithm also works for arbitrary RSA modules). We shall do this by giving a short procedure to obtain the value of $\varphi(N)$ given the order of a random element. \\

Recall that we will have an oracle (in particular the output of Shor's algorithm) that in principle will give a number which is the order of a random  element in $(\Z/N\Z)^*$. Hence to bound the probability,  given certain order, a divisor of $\varphi(N)$, we first need to control how many elements have this as its order. For that purpose we start by proving the following technical lemmata. The results are number theoretic analysis of the group $(Z/N\Z)^*$. First we will prove that the function $N(x)=\#\{a\in(\Z/N\Z)^*:\text{ord}_N(a)=x\}$ is multiplicative. This, will allow us to restrict our attention to prime powers, to find an explicit formula for the function $N(x)$ for any possible order $x$. The explicit formula is the key to give our bounds on the probability.  This will be the content of Theorem \ref{th:principal}.

\begin{lemma}\label{lem:ordmultgen}  Let $N$ be an integer, and consider the function  defined as  $N(x)=\#\{a\in(\Z/N\Z)^*:\text{ord}_N(a)=x\}$. Then 
$N(x)$ is a multiplicative function.
\end{lemma}
\begin{proof} Suppose $(x_1,x_2)=1$ are two relatively prime integers so that $N(x_1)=n_1$, and $N(x_2)=n_2$. We will prove that if    $C_{x_1}=\{a_1,\dots,a_{n_1}\}$ are the elements with order $x_1$ and $C_{x_2}=\{b_1,\dots,b_{n_2}\}$  are the elements with order $x_2$, then $C_{x_1,x_2}=\{a_ib_j:1\le i\le n_1,1\le j\le n_2\}$ contains the elements with  order $x_1x_2$.

\

Suppose ord$_N(a)=x$ and let $(x,y)=1$. Then ord$_N(a^y)=x$. Indeed suppose ord$_N(a^y)=z$. Since $(a^y)^x=(a^x)^y=1\pmod N$, $z|x$. But we know that  for some integers $u,v$ we have $a^z=(a^z)^{uy+vx}=(a^y)^{zu}(a^x)^{zv}=1$, so $x|z$ and they are the same.

\

Now, we see that every element in $C_{x_1,x_2}$ has order $x_1x_2$. First note that $(a_ib_j)^{x_1x_2}=1\pmod N$, so ord$_N(a_ib_j)|x_1x_2$ Suppose ord$_N(a_ib_j)=z$. Then $z=z_1z_2$ where $z_1=(z,x_1)$, and  $z_2=(z,x_2)$, but then 
$(a_ib_j)^{z_1z_2}=1\pmod N$ so $(a_i^{z_2})^{z_1}=((b_j)^{-1})^{z_1})^{z_2}$, but ord$_N(a_i^{z_2})=x_1$ and ord$_N(b_j)^{-1})^{z_1}=x_2$, so they can not be equal unless $z_1=x_1$ and $z_2=x_2$.
On the other hand $a_ib_j\ne a_Ib_J$ for any $(i,j)\ne (I,J)$, again for the same reason since otherwise $\frac{a_i}{a_I}=\frac{b_J}{b_j}$ which is impossible since ord$_N(\frac{a_i}{a_I})|x_1$ and ord$_N(\frac{b_J}{b_j})|x_2$. 

\

So we have proved that the $n_1n_2$ elements in $C_{x_1,x_2}$ are indeed different and have order $x_1x_2$. Now we need to prove that there are no more. Suppose $w\in(\Z/N\Z)^*$ has ord$_N(w)=x_1x_2$. We will prove that $w$ is the product of two elements $w=a_ib_j$ for $a_i\in C_{x_1}$ and $b_j\in C_{x_2}$. Suppose $u,v$ are integers such that $ux_1+vx_2=1$. Then  we have 
$$
w=w^{ux_1+vx_2}=w^{ux_1}w^{vx_2}
$$
But if ord$_N(w^{ux_1})=z$ then $z|x_2$ and $(u,z)=1$ so ord$_N(w^{x_1})=z$ and hence $z=x_2$ in the same way we prove that  ord$_N(w^{vx_2})=x_1$ finishing the proof of the lemma.
\end{proof}

\begin{lemma}\label{lem:ord} Let $N=pq$ the product of two primes, $x|[p-1,q-1]$ and  let $N(x)=\#\{a\in(\Z/N\Z)^*:\text{ord}_N(a)=x\}$.  Then
\begin{equation}\label{eq:orderspq}
N(x)=\varphi(x)(x,N-1)\sum_{\substack{d|(x,N-1)\\ \left(d,\frac{x}{(x,N-1)}\right)=1}}\frac{\mu^2(d)}{d}
\end{equation}
\end{lemma}
\begin{proof} We will denote $x$ as $x=\prod_{1\le i\le r}p_i^{\alpha_i}P_xQ_x$ where $P_x|P$ and $Q_x|Q$, with the notation in  (\ref{eq:factorpq}). Since $N(x)$ is multiplicative,  we just need to compute $N(l^a)$ for $l$ a prime number. Now if $l\nmid (p-1,q-1)$ and $l|p-1$, then if $e_p|(p-1)$ and $e_q|q-1$ are such that $[e_p,e_q]=l^a$, then $e_p=l^a$, $e_q=1$ and we have $N(l^a)=\varphi(l^a)$. If $l=p_i$ for some $1\le i \le r$, then  we have to distinguish two cases. 

\medskip

\underline{Case 1.} Suppose $\alpha_i\le m_i$. Then, if $e_p|(p-1)$ and $e_q|q-1$ are such that $[e_p,e_q]=p_i^{\alpha_i}$, then either 
$e_p=p_i^{\alpha_i}$ and $e_q=p_i^{c_i}$ for some $0\le c_i\le\alpha_i$ or $e_q=p_i^{\alpha_i}$  and $e_p=p_i^{c_i}$ for some $0\le c_i\le\alpha_i$. And, since we are counting $p_i^{\alpha_i}$ twice, we get in this case 
\begin{eqnarray*}
&&2\sum_{0\le c_i\le \alpha_i}\varphi(p_i^{\alpha_i})\varphi(p_i^{c_i})-\varphi(p_i^{\alpha_i})^2=2p_i^{\alpha_i}\varphi(p_i^{\alpha_i})-\varphi(p_i^{\alpha_i})^2\\
&&=\varphi(p_i^{\alpha_i})(2p_i^{\alpha_i}-p_i^{\alpha_i}+p_i^{\alpha_i-1})=\varphi(p_i^{\alpha_i})p_i^{\alpha_i}\left(1+\frac1{p_i}\right)\end{eqnarray*}
\medskip

\underline{Case 2.} Suppose $m_i<\alpha_i$ and let $a_i=M_i$.  Then, if $e_p|(p-1)$ and $e_q|q-1$ are such that $[e_p,e_q]=p_i^{\alpha_i}$, then $e_p=p_i^{\alpha_i}$ while $e_q=p_i^{c_i}$ for some $0\le c_i\le m_i$, so in this case
$$
N(p_i^{\alpha_i})=\varphi(p_i^{\alpha_i})\sum_{0\le c_i\le m_i}\varphi(p_i^{c_i})=\varphi(p_i^{\alpha_i})p_i^{m_i}
$$
Putting all together  we get 
\begin{eqnarray*}
N(x)&=&\varphi(P_xQ_x)\prod_{\substack{\alpha_i\le m_i}}\varphi(p_i^{\alpha_i})p_i^{\alpha_i}\left(1+\frac1{p_i}\right)
\prod_{\substack{\alpha_i>m_i}}\varphi(p_i^{\alpha_i})p_i^{m_i}\\
&=&\varphi(P_xQ_x)(x,N-1)\varphi(\frac{x}{P_xQ_x})\prod_{\alpha_i\le m_i}\left(1+\frac1{p_i}\right)\\
&=&\varphi(x)(x,N-1)\sum_{\substack{d|(x,N-1)\\ \left(d,\frac{x}{(x,N-1)}\right)=1}}\frac{\mu^2(d)}{d}.
\end{eqnarray*}
\end{proof}
\begin{lemma} Let $N=pq$  be the product of two prime numbers $p$ and $q$ and let $D|[p-1,q-1]$. Then, $(D,p-1,q-1)=(D,N-1)$.
\end{lemma}
\begin{proof} First we note that for any prime $l$ so that $l^a|[p-1,q-1]$, then either $l^a|p-1$ or $l^a|q-1$ and hence, noting that
$N-1=(p-1)q+(q-1)$,   if $l^a|([p-1,q-1],N-1)$, then $l^a|(p-1,q-1)$ and hence we deduce  $(D,N-1)|(D,p-1,q-1)$. On the other hand $(D,p-1,q-1)|( D,N-1)$, so they must be equal.
\end{proof}
\begin{corollary}\label{cor:biggerd} If $D|[p-1,q-1]$ then $D(D,N-1)|\varphi(N)$.
\end{corollary}
\begin{proof} Using the previous Lemma we see that  
$$
D(D,N-1)=D(D,(p-1,q-1))
$$
and since
$D|[p-1,q-1]$ and $(D,(p-1,q-1))|(p-1,q-1)$ we get 
$$
D(D,N-1)|[p-1,q-1](p-1,q-1)=
(p-1)(q-1)=\varphi(N).
$$
\end{proof}
\begin{corollary}\label{cor:pq} Let $p-1=\prod_{i=1}^rp_i^{a_i}P$, $q-1=\prod_{i=1}^rp_i^{b_i}Q$ where a prime $l|(p-1,q-1)$ if and only if $l=p_i$ for some $1\le i\le r$.  Suppose $D|[p-1,q-1]$ and  consider the sequence $D_1=D$, $D_{j+1}=\frac{D_j}{(D_j,N-1)}$. Then for some $j_0$ we have 
$D_j=D_{j_0}$ for all $j\ge j_0$ and $D_{j_0}|PQ$.
\end{corollary}
\begin{proof} Let $m_i=\min\{a_i,b_i\}$ and suppose $p_i^{\alpha_i}||D_j$ for $1\le i\le r$. Then, if $\alpha_i\le m_i$ we have
$p_i^{\alpha_i}||(D_j,N-1)$ and $p_i\nmid D_{j+1}$. If $\alpha_i>m_i$, then $p_i^{m_i}||(D_j,N-1)$ and $p_i^{\alpha_i-m_i}||D_{j+1}$ and iterating we get the result.
\end{proof}
\begin{remark} Observe that if $D=\prod_{i=1}^rp_i^{\alpha_i}P_DQ_D=C_DP_DQ_D$, where $P_D|P$ and $Q_D|Q$, then 
$D_{j_0}=P_DQ_D$ and $\frac{D}{D_{j_0}}=\prod_{i=1}^rp_i^{\alpha_i}$.
\end{remark}

One can use the previous remark to give a result that allows factorization directly in terms of $(p-1,q-1)$, improving Theorem \ref{th:gcd} in some cases.
\begin{theorem} Let $N=pq$ where $p$ and $q$ and $D=C_DP_DQ_D=C_DD_{j_0}$ are as in Corollary \ref{cor:pq} and denote $F=(C_D,N-1)$. Then, if 
$$
D_{j_0}>\frac{(p-1+q-1)}{F^2},
$$
we can factor $N$ in polynomial time.
\end{theorem}
\begin{proof}
The proof is the same as in Theorem \ref{th:squareroot}. Note that 
$$
N-1-(p-1+q-1)=\varphi(N)=[p-1,q-1](p-1,q-1)
$$
and since $D|[p-1,q-1]$, and $F|(p-1,q-1)$, we have
$$
N-1-(p-1+q-1)=DFX,
$$
for some integer $X$. Dividing by $DF$ we get 
$$
\frac{N-1}{DF}=X+\frac{p-1+q-1}{DF},
$$
but, by hypothesis,
$$
DF=D_{j_0}C_DF\ge D_{j_0}F^2>(p-1+q-1)
$$
and, in particular $X=\left[\frac{N-1}{DF}\right]$. Once we have $X$, we have $\varphi(N)$ and we can factor $N$.
\end{proof}

We are now ready to prove the main theorem. We need the following:

\begin{definition}   Let $O$ be the  random  oracle defined in the following way:  each time we  call $O$, it selects  uniformly at random a pair   from the set $C=\{(a,d): a\in(\Z/N\Z)^* ,d|[p-1,q-1], \text{ord}_N(a)=d\}$ , and returns $d$.
\end{definition}
\begin{theorem} \label{th:principal}  Let $\epsilon>0$, $K$ and explicit constant depending on $\epsilon$ and $L>K$ a positive integer. Let  $p<q$ be prime numbers of $L$ bits  and $N=pq$. We can compute $\varphi(N)$ and factor $N$ in polynomial time with probability at least $1-\frac{1}{N^{1/2-\epsilon}}$,  with just one call to $O$. 
\end{theorem}
\begin{proof} Let $x$ be the integer returned by $O$. Note that if we denote $P(x)$ the probability that the oracle returns $x$
then, since the oracle selects $a$ and $d$ uniformly at random, we have $P(x)=\frac{N(x)}{\varphi(N)}$.

\

 If $x\ge 4N^{\frac12}$, recalling that  $x|[p-1,q-1]|\varphi(N)$, we just apply Theorem \ref{th:squareroot} to get the factorization. So we suppose $x< 4N^{1/2}$. If, moreover,  $N(x)<4N^{1/2}$, $O$ will select $x$ with probability smaller than $\frac{16}{N^{1/2}}$, by noticing that 
 $$
 \varphi(N)=N-(p+q)+1\ge N-4\sqrt  N>CN,
 $$
for  $C=16$  if  $N> 30$, since then we have $N-4\sqrt  N>N/4$.  Adding in all possible $x|[p-1,q-1]$ we get 
$$
\sum_{x|[p-1,q-1]}\frac{N(x)}{\varphi(N)}<\frac{16}{N^{1/2}}\tau([p-1,q-1]),
$$
where $\tau(\cdot)$ denotes the divisor function. Using Theorem 4 of \cite{jorge}, we see that 
$$
\tau([p-1,q-1])<[p-1,q-1]^{3/\log\log[[p-1,q-1]]}<N^{3/\log\log N},
$$
since $[p-1,q-1]<\varphi(N)<N$, which gives the stated result for  $K=\frac{e^{\frac{3}{\epsilon}}}{\log 2}$.

\

Otherwise we may assume $N(x)\ge 4\sqrt{N}$. In this case we consider the integer  
$x(x,N-1)$ which is a divisor of $\varphi(N)$ by Corollary \ref{cor:biggerd}, and by Lemma \ref{lem:ord}, 
\begin{eqnarray*}
4\sqrt N&\le& N(x)=\varphi(x)(x,N-1)\sum_{\substack{d|(x,N-1)\\ \left(d,\frac{x}{(x,N-1)}\right)=1}}\frac{\mu^2(d)}{d}\\
&\le& x(x,N-1)\prod_{p|(x,N-1)}\left(1-\frac{1}{p^2}\right)<x(x,N-1),
\end{eqnarray*}
where we have used 
$$
\varphi(x)=x\prod_{p|x}\left(1-\frac1p\right)\le x\prod_{p|(x,N-1)}\left(1-\frac1p\right),
$$ and 
$$
\sum_{\substack{d|(x,N-1)\\ \left(d,\frac{x}{(x,N-1)}\right)=1}}\frac{\mu^2(d)}{d}\le \sum_{d|(x,N-1)}\frac{\mu^2(d)}{d}=\prod_{p|(x,N-1)}\left(1+\frac1p\right).
$$
So,  we can apply Theorem \ref{th:squareroot} with $D=x(x,N-1)$ to get the factorization. As remarked in Corollary \ref{th:laphi}, the value of $\varphi(N)$ is also determined.
\end{proof}

\section{The Algorithm}\label{sec:alg}
We have already proved in the previous theorem that given an RSA module $N$ which is the product of two primes of the same length, and a number $x$ which is the order modulo $N$ of a random number (as in the output of Shor's quantum order finding algorithm) there is a polynomial time algorithm that computes $\varphi(N)$ which very high probability. Let us put together the different steps of this algorithm, taken from theorems \ref{th:squareroot} and \ref{th:principal}. \\

{\bf Algorithm:} \\
Input: An RSA modulo $N$ which is the product of two primes of the same length, and a positive integer $x$ which is the order modulo $N$ of a random number. \\
Output: $\varphi(N)$. \\
$(1)$ Compute the greatest common divisor $w: = (x, N-1)$\\
$(2)$ Compute the product $D = x \cdot w$ \\
$(3)$ Compute the quotient $X$ of the division $(N+1) / D$ \\
$(4)$ Compute the product $X \cdot D$. \\

As proved in theorems \ref{th:squareroot} and \ref{th:principal}, the result of step $(4)$ gives the value of $\varphi(N)$ with very high probability. The algorithm that we propose consists thus in the computation of a greatest common divisor, two products, and a computation of a quotient in a division with reminder. After the proof of Harvey  and van der Hoeven in \cite{multi} of the $1971$ conjecture of Schönhage and Strassen  that the product of two $n$-bit integers can be computed in quasilinear time, namely $O(n\log n$) bit operations,  the complexity of the algorithm corresponds to the computation of the greatest common divisor which  is of order $O(\log N \cdot (\log \log N)^a)$, where we can take as exponent any real number $a$ greater than $2$ (see \cite{mo}). 

\

{\bf Acknowledgements.} We would like to thank the anonymous referees for pointing out some mistakes contained  in a previous version, and for  their valuable comments improving the final presentation of the paper. 
\bibliography{factoring}

\bibliographystyle{plain}
\end{document}